\documentclass[11pt]{article}

\usepackage[english]{babel}
\usepackage[utf8]{inputenc}
\usepackage[authoryear]{natbib}
\usepackage[letterpaper,margin=1in]{geometry}
\usepackage{setspace}
\usepackage{amsmath,amssymb,amsthm}
\usepackage{graphicx}
\usepackage{float}
\usepackage{booktabs}
\usepackage{siunitx}
\usepackage{multirow}
\usepackage[colorlinks=true, allcolors=blue]{hyperref}

\newtheorem{theorem}{Theorem}[section]
\newtheorem{lemma}[theorem]{Lemma}
\newtheorem{assumption}[theorem]{Assumption}
\theoremstyle{definition}

\title{A Wide-Sense Stationarity Test Based on the Geometric Structure of Covariance}
\author{
  Yinbu Wang$^{1}$\thanks{Email: wamgyb@mail.nwpu.edu.cn} \and
  Yong Xu$^{1}$\thanks{Email: hsux3@nwpu.edu.cn}
}

\date{%
  $^{1}$School of Mathematics and Statistics, Northwestern Polytechnical University, Xi'an, China
}

\begin{document}
\maketitle

\begin{abstract}
This paper presents a test for wide-sense stationarity based on the geometry of the covariance function. We estimate local patches of the covariance surface and then check whether the directional derivative in the $(1,1,0)$ direction is zero on each patch. The method only requires the covariance function to be locally smooth and does not assume stationarity in advance. It can be applied to general stochastic dynamical systems and provides a time-resolved view. We apply the test method to an SDOF system and to a Duffing oscillator. These examples show that the method is numerically stable and can detect departures from stationarity in practice. 
\end{abstract}

\section{Introduction}
Wide-sense stationarity (WSS) is a commonly used assumption in time series and stochastic process studies, and it serves as the basis for many frequency domain methods and statistical inference (\cite{brockwell2009time,liang2015random}). In engineering signals, random vibration, and stochastic dynamical systems, the response is often governed by nonlinear dynamics and driven by non-stationary external excitation, and system parameters may also vary over time. Therefore, to justify frequency domain analysis, it is necessary to evaluate whether the process is WSS and, if not, to locate the time intervals over which a WSS approximation remains reasonable(\cite{lin2004probabilistic,sun2006stochastic,forgoston2018primer}).
\newline
For non-stationary processes, existing work follows two lines. The first is the locally stationary model, which treats non-stationarity as a slow time variation and builds tests and measures based on this idea(\cite{priestley1965evolutionary,dahlhaus1997fitting,paparoditis2009testing,dette2011measure,nason2013test,aue2020testing,van2021nonparametric}). The second is the piecewise stationary model, which views non-stationarity as changes at a few time points or as piecewise changes, and develops a range of tests and estimation methods(\cite{page1954continuous,hinkley1971inference,andrews1993tests,yao1987approximating,bai1998estimating,einmahl2003empirical}); it can also be extended to handle cases where slow variation and sudden shifts occur together(\cite{last2008detecting,casini2024change}). A shared problem is that one usually needs to decide whether the main source of non-stationarity is slow variation or sudden shifts. If this choice is wrong, it can introduce bias and may even lead to misleading conclusions (\cite{casini2024change}). This issue is more common in stochastic dynamical systems: Noise can cause sudden switches between metastable states, so the structure may change abruptly rather than slowly \cite{forgoston2018primer}.On the other hand, putting an SDE directly into a change-point framework often requires a more explicit model setup, so it may be less reliable for complex systems.
\newline
This work builds on a fact from differential geometry \cite{do2016differential}: if a stochastic process is WSS, its covariance surface has a cylindrical structure, with generator direction $(1,1,0)$. We therefore assess WSS by testing the local change of the covariance surface along this direction. Using local polynomial regression \cite{fan2018local,masry1996multivariate,de2013derivative}, we estimate the required derivatives and test whether the directional derivative along $(1,1,0)$ is zero. The proposed method does not rely on local spectra or a specific parametric time series model. It only requires the covariance surface to be continuous and differentiable. When model information is limited, or the system is complex, this geometric approach provides a direct and easy to interpret test. The framework is especially suitable for stochastic differential equations and stochastic dynamical systems. On the one hand, under common conditions, solutions of SDEs have well-behaved properties. On the other hand, it is common to obtain multiple independent sample paths through numerical simulation or repeated experiments, which helps stabilize covariance estimation and support local diagnostics.
\newline
Section 2 describes the geometric structure of the covariance function for WSS processes and introduces local cylindrification as the key step for constructing the WSS test. Section 3 builds the test statistic via local polynomial regression and establishes its properties. Section 4 evaluates the method on two stochastic dynamical SDOF oscillators and a stochastic Duffing oscillator and further compares its performance with the stationarity measure of \cite{dette2011measure}.

\section{Geometric Shape of the Covariance Function}
Let $X_t$ be a stochastic process on an open interval $T$ with finite second moments, and let $r(s,t)$ be its covariance function
\begin{equation}
    r(s,t)=E(X_tX_s)-E(X_s)E(X_t).
\end{equation}
We say that $X_t$ is wide-sense stationary on $T$ if there exists a positive definite function $h:T \to R$ such that
\begin{equation}
r(s,t)=h(s-t),\qquad E(X_t)=m,
\end{equation}
where $m$ is a constant. The following lemma gives an equivalent form for the covariance part. From now on, we only consider the case $E(X_t)=0$.
\begin{lemma}\label{lemma1}
The process $X_t$ is WSS on a connected open set $T$ with zero mean if and only if its covariance $r(s,t)\in C^1(T\times T)$ and
\begin{equation}
  r_s + r_t = 0.
\end{equation}
Here
\begin{equation}
    \begin{split}
        &r_s=\frac{\partial r(s,t)}{\partial s},\\
        &r_t=\frac{\partial r(s,t)}{\partial t}.
    \end{split}
\end{equation}
\end{lemma}
\begin{proof}
If $X_t$ is WSS on $T$, then $r(s,t)=h(s-t)$ for some $h$. Hence
\begin{equation}
\begin{split}
    r_s+r_t&=\frac{\partial h}{\partial (s-t)}\frac{\partial (s-t)}{\partial s}+\frac{\partial h}{\partial (s-t)}\frac{\partial (s-t)}
    {\partial t}\\
    &=\frac{\partial h}{\partial (s-t)}\Bigl(\frac{\partial (s-t)}{\partial s}+\frac{\partial (s-t)}
    {\partial t}\Bigr)\\
    &=0.
\end{split}
\end{equation}
Assume $r\in C^1(T\times T)$ and $r_s + r_t = 0$ everywhere. Let
\begin{equation}
    u = s - t,\quad v = t,
\end{equation}
and define
\begin{equation}
h(u,v) = r(u + v,v). 
\end{equation}
By the chain rule,
\begin{equation}
\begin{split}
        \frac{\partial h(u,v)}{\partial v}&=\frac{\partial r(s(u,v),t(u,v))}{\partial s}\frac{\partial s}{\partial v}+\frac{\partial r(s(u,v),t(u,v))}{\partial t}\frac{\partial t}{\partial v}\\
        &=\frac{\partial r(s(u,v),t(u,v))}{\partial s}+\frac{\partial r(s(u,v),t(u,v))}{\partial t}\\
        &=0.
\end{split}
\end{equation}
Thus $h$ does not depend on $v$. Write $h(u,v)=h(u)$, and we obtain $r(s,t)=h(s-t)$, which is exactly the form of a WSS process.
\end{proof}
Lemma~\ref {lemma1} can be regarded as a partial differential equation description of WSS. The idea is not new. In what follows, we extend it to the surface given by the covariance function and obtain a geometric interpretation of WSS.

\subsection{Geometric Shape}
We first see how WSS affects the surface generated by $r(s,t)$. The following argument relies on some standard results in differential geometry.
\newline
For a smooth function $f(x,y)$ the Gaussian curvature at any point is
\begin{equation}
    K=\frac{f_{xx}f_{yy}-f_{xy}^2}{\bigl(1+f_x^{\,2}+f_y^{\,2}\bigr)^{2}}.
\end{equation}
Consider the surface $M=\left\{x=(s,t,r(s,t))|s,t\in T\right\}$ which is simply the graph of $r(s,t)$. Based on classical results in differential geometry \cite{do2016differential}, we have the following lemma. Here we assume $r(s,t) \in C^2(T \times T)$ so that the Gaussian curvature is well defined.
\begin{lemma}\label{lemma2}
    If $X_t$ is WSS, then $M$ is a cylindrical surface, and its ruling direction is $(1,1,0)$.
\end{lemma}
\begin{proof}
For a WSS process, we have
\begin{equation}
r(s,t)=h(s-t),
\end{equation}
which implies
\begin{equation}
  r_s+r_t = 0 \quad \text{for all }(s,t)\in T\times T.  
\end{equation}
Take any $(s_0,t_0)\in T\times T$ and consider the curve
\begin{equation}
\gamma(\lambda)=(s_0+\lambda,t_0+\lambda,r(s_0+\lambda,t_0+\lambda)).
\end{equation}
Then
\begin{equation}
    \begin{split}
        &\gamma'(\lambda)=(1,1,r_s+r_t)=(1,1,0),\\
        &\gamma''(\lambda)= 0.
    \end{split}
\end{equation}
Hence $\gamma$ is a straight line whose direction vector is $v=(1,1,0)$. If we let $(s_0,t_0)$ change, we obtain a family of parallel straight lines that cover $M$. Therefore, $M$ is a ruled surface whose rulings are all parallel to $v$. Let $f(s,t)=h(s-t)=h(l)$, then
\begin{equation}
f_s=h'(l),\quad f_t=-h'(l),\quad f_{ss}=f_{tt}=h''(l),\quad f_{st}=-h''(l).
\end{equation}  
The Gaussian curvature is
\begin{equation}
K=\frac{f_{ss}f_{tt}-f_{st}^{2}}{\bigl(1+f_s^{2}+f_t^{2}\bigr)^{2}}=\frac{h''(l)h''(l)-(-h''(l))^{2}}{(1+2h'(l)^{2})^{2}}=0.
\end{equation}
Thus $M$ is a developable ruled surface. Since all rulings are parallel to the same direction vector, the only possible case is that $M$ is a cylindrical surface with ruling direction $\mathbf v$.
\end{proof}
\begin{figure}[H]
    \centering
    \includegraphics[width=0.5\linewidth]{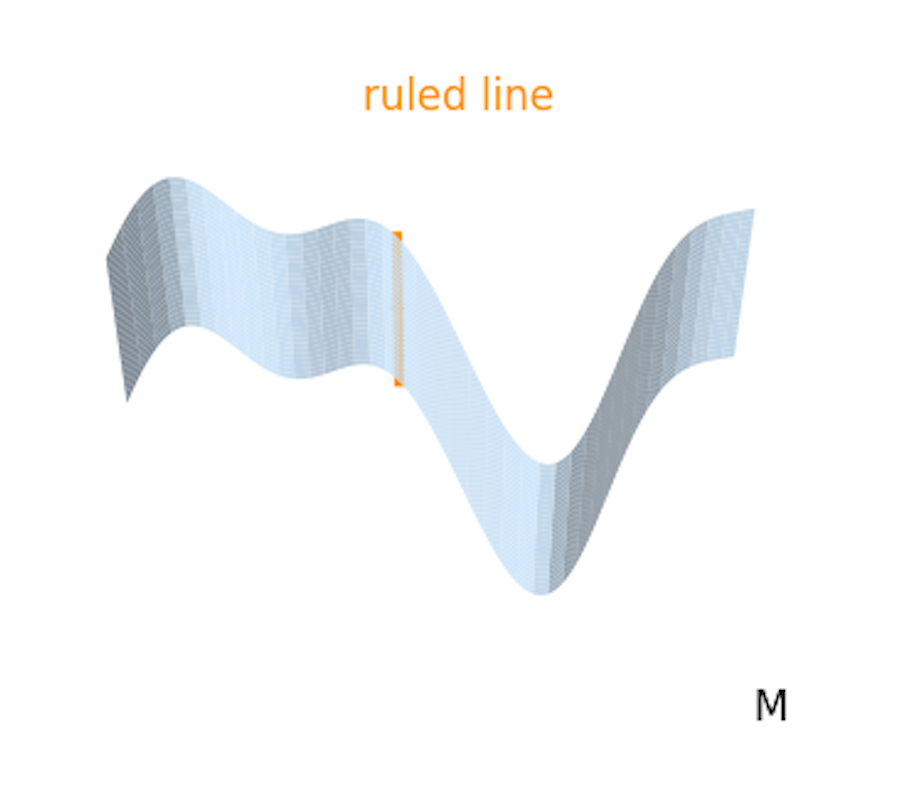}
    \caption{A surface $M$ generated by straight rulings with the same direction.}
    \label{M}
\end{figure}
Figure~\ref{M} shows such a surface. If the process is WSS, the ruling direction is $(1,1,0)$. We will use this property to build a test for WSS.

\subsection{Local Cylindrification}
Lemma~\ref{lemma2} shows that there is a direct link between WSS and the geometric structure of the covariance surface. Thus, testing WSS can be turned into checking whether the surface is a cylinder with ruling direction $(1,1,0)$. Interestingly, this detection method can also be applied locally. The following lemma provides a construction of this approach.
\begin{lemma}\label{lemma3.5}
    Assume $r\in C^1(T\times T)$, $T<\infty$, and $T\times T=\bigcup_{i=1}^{M}\Omega_i$ with $\Omega_i\cap \Omega_j=\emptyset$ when $i\neq j$. Each $\Omega_i$ is a convex domain and the mean of $r$ on $\Omega_i$ is
    \begin{equation}
        \bar r_{\Omega_i} = \frac{1}{|\Omega_i|}\int\int_{\Omega_i} r(s,t)dsdt.
    \end{equation}
    Define the local cylindrification approximation 
    \begin{equation}
        r^{\mathrm{cyl}}_h(s,t) = \sum_{i=1}^M r^{\mathrm{cyl}}_i\, \mathbf \chi_{\Omega_i},
    \end{equation}
    where 
    \begin{itemize}
        \item $\chi_{\Omega_i}$ is the indicator function,
        \item $\sup_{i\leq M} \text{diam}(\Omega_i)\leq h$,
        \item $r^{\mathrm{cyl}}_i$ is a cylindrical surface function.
    \end{itemize}
    Then there exists such an $r^{\mathrm{cyl}}_h(s,t)$ for which
    \begin{equation}
    \lim_{h\to 0}\|r-r^{\mathrm{cyl}}_h\|_{L^2(T\times T)}\to 0.
    \end{equation}
\end{lemma}
\begin{proof}
    First take $r^{\mathrm{cyl}}_i=\bar r_{\Omega_i}$. By the Poincaré inequality, we have in each $\Omega_i$
    \begin{equation}
    \begin{split}
 \|r-r^{\mathrm{cyl}}_h\|_{L^2(\Omega_i)}&=     \|r-\bar r_{\Omega_i}\|_{L^2(\Omega_i)}\\
 &\leq C \text{diam}(\Omega_i) \|\nabla r\|_{L^2(\Omega_i)}\\
 &\leq Ch \|\nabla r\|_{L^2(\Omega_i)}.
    \end{split}
    \end{equation}
    Since $ r \in C^1(\Omega_i)$ and $ \int_{\Omega_i} (r(x)-\bar r_{\Omega_i})dx = 0$, we can apply the result of \cite{bebendorf2003note}. Let $k=\sup_{x\in T\times T}|\nabla r(x)|$. Then
    \begin{equation}
    \begin{split}
        \|\nabla r\|_{L^2(\Omega_i)}&=(\int_{\Omega_i}|\nabla r(x)|^2dx)^{\frac{1}{2}}\\
        &\leq (\int_{\Omega_i}|k|^2dx)^{\frac{1}{2}}\\
        &\leq (|k|^2|\Omega_i|)^{\frac{1}{2}}.
    \end{split}
    \end{equation}
    Hence
    \begin{equation}
        \|r-r^{\mathrm{cyl}}_h\|_{L^2(\Omega_i)}\leq  Chk|\Omega_i|^{\frac{1}{2}}.
    \end{equation}
    Constant functions are included in the class of cylindrical surface functions. Thus, there exists a family of cylindrical surface functions such that on each $\Omega_i$ we have
    \begin{equation}
         \|r-r^{\mathrm{cyl}}_h\|_{L^2(\Omega_i)}\leq  Chk|\Omega_i|^{\frac{1}{2}}.
    \end{equation}
    We now show that the above relation holds on the whole domain. Since $\Omega_i$ and $\Omega_j$ do not overlap we obtain
    \begin{equation}
    \begin{split}
        \|r-r^{\mathrm{cyl}}_h\|^2_{L^2(T\times T)}&=\sum_{i=1}^M\|r-r^{\mathrm{cyl}}_h\|^2_{L^2(\Omega_i)}\\
        &\leq C^2k^2 h^2\sum_{i=1}^M|\Omega_i|\\
        &=C^2k^2h^2|T\times T|.
    \end{split}
    \end{equation}
    Therefore
    \begin{equation}
        \lim_{h\to 0}\|r-r^{\mathrm{cyl}}_h\|_{L^2(T\times T)}\leq \lim_{h\to 0}\ C^2k^2h^2|T\times T|\to 0.
    \end{equation}
\end{proof}
\begin{figure}[H]
    \centering
    \includegraphics[width=0.5\linewidth]{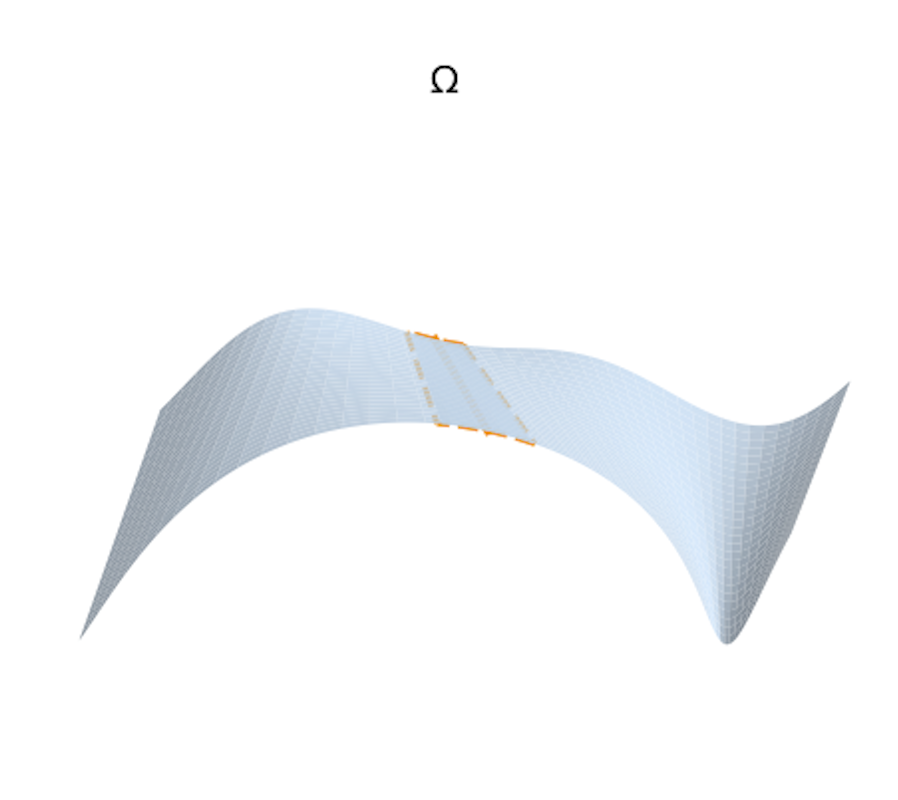}
    \caption{Local cylindrification on a surface patch $\Omega$ with parallel rulings and a unique ruling direction.}
    \label{Local cylindrification}
\end{figure}
The lemma above shows that for the covariance surface $r$ we can build a locally cylindrical approximation $r^{\mathrm{cyl}}$ such that $r^{\mathrm{cyl}}$ converges to $r$ in $L^{2}$. As in Figure~\ref{Local cylindrification}, any surface that satisfies the assumption contains a patch $\Omega$ on which it is locally cylindrical.
\newline
Note that Lemma~\ref{lemma3.5} only establishes the existence of a local cylindrification approximation; it does not specify the ruling direction on each patch. When the original function is a cylindrical surface, the ruling direction produced by this construction with the true ruling direction. However, the original function is not cylindrical and therefore does not possess a ruling in the classical sense. This motivates a stronger construction, built on Lemma~\ref{lemma3}, which both selects an appropriate 'generating direction' and remains applicable to non-cylindrical surfaces.

\begin{lemma}\label{lemma3}
Assume that $r\in C^{2}(T\times T)$, $T<\infty$, we keep the same setting as in Lemma~\ref{lemma3.5}.
For each $i$, choose a point $x_i=(s_i,t_i)\in \Omega_i$ and define the local first-order Taylor polynomial
\begin{equation}
    r^{\mathrm{cyl}}_i(s,t)= r(x_i)+ r_s(x_i)(s-s_i)+r_t(x_i)(t-t_i),\qquad (s,t)\in\Omega_i.
\end{equation}
Define the local cylindrification approximation
\begin{equation}
  r^{\mathrm{cyl}}_h(s,t)=\sum_{i=1}^{M} r^{\mathrm{cyl}}_i(s,t)\,\chi_{\Omega_i}(s,t),
\end{equation}
Then there exists a constant $C>0$, independent of $h$, such that
\begin{equation}
  \|r-r^{\mathrm{cyl}}_h\|_{L^2(T\times T)}\leq C h^{2},
\end{equation}
and
\begin{equation}
    \lim_{h^2\to 0}\|r-r^{\mathrm{cyl}}_h\bigr\|_{L^2(T\times T)} \to 0.
\end{equation}
\end{lemma}
\begin{proof}
Fix $i\in\{1,\dots,M\}$ and let $x_i=(s_i,t_i)\in\Omega_i$. Since $r\in C^{2}(T\times T)$ and $\Omega_i$ is convex with $\text{diam}(\Omega_i)\le h$, Taylor's theorem with remainder gives, for
every $x=(s,t)\in\Omega_i$,
\begin{equation}
r(x)= r^{\mathrm{cyl}}_i(x) + R_i(x).
\end{equation}
Let
\begin{equation}
      M=\max_{x_i\in T\times T}\{|r_{ss}(x_i)|,|r_{tt}(x_i)|,|r_{st}(x_i)|\},
\end{equation}
then
\begin{equation}
    |R_i(x)|\leq \frac{1}{2}(Mh^2+2Mh^2+Mh^2)\leq C_0h^2,
\end{equation}
for some constant $C_0=2M$ depending only on $r$. It follows that
\begin{equation}
    \|r-r^{\mathrm{cyl}}_i\|_{L^{2}(\Omega_i)}
  = ( \int_{\Omega_i} |R_i(x)|^{2}\,dx )^{1/2}
  \le C_0 h^{2}\,|\Omega_i|^{1/2}.
\end{equation}
Since the interiors of the $\Omega_i$ are disjoint and $T\times T = \bigcup_{i=1}^{M} \Omega_i$, we obtain
\begin{equation}
    \begin{split}
     \|r-r^{\mathrm{cyl}}_h\|_{L^{2}(T\times T)}^{2}
     &=\sum_{i=1}^{M}\|r-r^{\mathrm{cyl}}_i\|_{L^{2}(\Omega_i)}^{2} \\
     &=\sum_{i=1}^{M} C_0^{2} h^{4}|\Omega_i|\\
     &= C_0^{2} h^{4} \sum_{i=1}^{M} |\Omega_i|\\
     &= C_0^{2} h^{4} |T\times T|.
    \end{split}
\end{equation}
Thus
\begin{equation}
      \|r-r^{\mathrm{cyl}}_h\|_{L^2(T\times T)} \leq C h^{2},
\end{equation}
where $C=\sqrt{|T\times T|}C_0$ and
\begin{equation}
    \lim_{h^2\to 0}\|r-r^{\mathrm{cyl}}_h\bigr\|_{L^2(T\times T)} \to 0.
\end{equation}
\end{proof}
Lemma 2.4 provides a more concrete construction, namely approximating the original covariance function by a local first-order Taylor expansion. For each patch $\Omega_i$, at the point $x_i$, the gradient of the Taylor plane coincides exactly with that of the original function, so the directional derivative $r_s+r_t $along the $(1,1)$ direction at $x_i$ agrees perfectly with the original function. This planar approximation does not require the original surface to be a cylindrical surface, and thus applies equally well to general non-cylindrical covariance surfaces.
\newline
At the same time, the Taylor plane provides, in a certain sense, a local “extension” mechanism: it suffices to evaluate $r_s + r_t$ at a single reference point in $\Omega_i$, and we can then construct an approximation to the directional derivative along $(1,1)$ over the entire patch. Moreover, the planar approximation based on the first-order Taylor expansion naturally aligns with the local polynomial regression approach adopted later in Section~3, which is an additional advantage of this construction.
\newline
It is worth noting that the piecewise planar approximation $r_h^{\mathrm{cyl}}$ is smooth within each patch, but its derivatives may be discontinuous across patch boundaries, so it is no longer globally $C^1$. We will show later that this lack of global smoothness does not affect the statistical inference based on $L^2$ error and directional derivatives. In what follows, the term “local cylindrification approximation” will refer to this approximation scheme.

\subsection{WSS Testing}
We now formalize the patched covariance surface and present a stationarity test for it. We first look at a surface that comes from two locally cylindrical patches. From here on, the covariance functions $r$ under consideration are the patched functions constructed by Lemma~\ref{lemma3}. For such functions, we need an appropriate description of the function space.
\begin{assumption}\label{ass1}
Let $\Omega := T\times T$ with $T<\infty$, and let $\Omega_1, \Omega_2 \subset \Omega$ be disjoint open sets with piecewise $C^1$ boundaries. Set
\begin{equation}
    \Sigma = \partial\Omega_1 \cap \partial\Omega_2,
\end{equation}
and assume that $\Sigma$ is a curve of Lebesgue measure zero. Let $2 < p < \infty$. Assume that the covariance surface $r : \Omega \to \mathbb{R}$ satisfies
\begin{enumerate}
  \item $r \in W^{1,p}(\Omega)$,
  \item
  $
    r|_{\Omega_k} \in W^{2,p}(\Omega_k)  \quad M_k := \{(s,t,r(s,t)) : (s,t)\in \Omega_k,k=1,2\}\text{ is a cylindrical surface},
  $
  \item across $\Sigma$ the first order derivatives of $r$ may have jumps.
\end{enumerate}
\end{assumption}
This means that $r$ is globally in $W^{1,p}(\Omega)$. On each locally cylindrical patch $\Omega_k$, it belongs to $W^{2,p}$, and structural changes are captured by joining the patches along $\Sigma$.
\begin{theorem}\label{the}
   Under Assumption~\ref{ass1}, the process $X_t$ is WSS on $T$ if and only if $r_s + r_t = 0$  almost everywhere holds on $\Omega_1$ and on $\Omega_2$.
\end{theorem}
\begin{proof}
If $X_t$ is WSS and $r|_{\Omega_k} \in W^{2,p}(\Omega_k)$ with $2 < p < \infty$, then Lemma~\ref{lemma1} tells us that $r_s + r_t = 0$ everywhere. Hence, it also holds on $\Omega_1$ and on $\Omega_2$.
\newline
Assume that $r_s + r_t = 0$ holds on $\Omega_1$ and on $\Omega_2$. We distinguish two cases.
\newline
\textbf{Case without first order derivative jumps.} Take $x_n \in \Omega_1$ and $z\in\Sigma$. Since $r$ is continuous and there is no first-order term, we have
\begin{equation}
    \lim_{x_n\to z} r_s(x_n)=r_s(z), \quad \lim_{x_n\to z} r_t(x_n)=r_t(z). 
\end{equation}
Then
\begin{equation}
    r_s(z)+r_t(z)=\lim_{x_n\to z}[r_s(x_n)+r_t(x_n)]=0.
\end{equation}
So $r_s + r_t = 0$ on $T$.
\newline
\textbf{Case with first order derivative jumps.} In this case, classical partial derivatives cannot be defined on $\Sigma$. We therefore use weak derivatives \cite{leoni2017first}. Since $r \in W^{1,p}(T \times T)$ with $p > 2$, the weak partial derivatives $\frac{\partial r}{\partial s}$ and $\frac{\partial r}{\partial t}$ satisfy for any $\phi \in C_c^\infty(T \times T)$
\begin{equation}
    \int_{T \times T} r \frac{\partial \phi}{\partial s}dsdt=-\int_{T \times T} \frac{\partial r}{\partial s}\phi ds dt,
\end{equation}
and
\begin{equation}
\int_{T \times T} r\frac{\partial \phi}{\partial t} dsdt=-\int_{T \times T} \frac{\partial r}{\partial t}\phi ds dt.
\end{equation}
Using integration by parts, we obtain
\begin{equation}
   \int_{T \times T} r \frac{\partial \phi}{\partial s} ds dt = -\int_{\Omega_1} \frac{\partial r}{\partial s} \phi ds dt - \int_{\Omega_2} \frac{\partial r}{\partial s} \phi ds dt + \int_{\Sigma} (r_{\Omega_1} \phi n_{s,1} + r_{\Omega_2} \phi n_{s,2}) d\sigma.
\end{equation}
Here $n_{s,1}$ and $n_{s,2}$ are the components in the $s$ direction of the outward unit normal vectors to $\Omega_1$ and $\Omega_2$ on $\Sigma$ and $d\sigma$ denotes integration on $\Sigma$. Since $r\in W^{1,p}(\Omega)$ we have $r_{\Omega_1}\phi n_{s,1} + r_{\Omega_2}\phi n_{s,2}=r_{\Omega_1}\phi(n_{s,1}+n_{s,2})$ which is bounded and because $n_{s,1}=-n_{s,2}$ the boundary term vanishes
\begin{equation}
    \int_{\Sigma} (r_{\Omega_1}\phi n_{s,1} + r_{\Omega_2}\phi n_{s,2}) d\sigma = \int_{\Sigma}r_{\Omega_1}\phi(n_{s,1}+n_{s,2})d\sigma=0.
\end{equation}
Performing the same operation for the $t$ direction and putting the two identities together, we get
\begin{equation}
    \int_{T \times T} r \left( \frac{\partial \phi}{\partial s} + \frac{\partial \phi}{\partial t} \right)  ds  dt = -\int_{\Omega_1} \left( \frac{\partial r}{\partial s} + \frac{\partial r}{\partial t} \right) \phi  ds  dt - \int_{\Omega_2} \left( \frac{\partial r}{\partial s} + \frac{\partial r}{\partial t} \right) \phi dsdt.
\end{equation}
Since $r_s + r_t=0$ holds on both $\Omega_1$ and $\Omega_2$ we have
\begin{equation}
     \int_{T \times T} r \left( \frac{\partial \phi}{\partial s} + \frac{\partial \phi}{\partial t} \right)ds dt=0.
\end{equation}
Thus in the weak sense $r_s + r_t = 0$ holds on $T\times T$.
\newline
Rewriting Eq.(44) under the change of variables Eq.(6)–(7) in Lemma 2.1, one obtains that $\tilde r(u,v)=r(u+v,v)$ is constant in the v-direction. Hence $\tilde r(u,v)=h(u)$ and therefore $r(s,t)=h(s-t)$.
\end{proof}
In this framework, the relation between testing $r_s + r_t = 0$ and testing WSS has been made explicit. Compared with other approaches, this method asks for only weak smoothness assumptions on the covariance function, but still keeps the test fully equivalent to WSS.

\section{Local Polynomial Regression}
Let $\{(\mathbf X_i,Y_i)\}_{i=1}^n$ be data with
$\mathbf X_i\in  R^{d}$ and model
\begin{equation}
    Y_i=m(\mathbf X_i)+\varepsilon_i,
\end{equation}
where $m$ is an unknown smooth function. Local polynomial regression approximates $m$ in a neighborhood of $\mathbf x_0$ by a Taylor polynomial of degree $p$
\begin{equation}
m(\mathbf x)\approx \sum_{j\le p} \beta_{j}(\mathbf x_0)(\mathbf x-\mathbf x_0)^{j},
\qquad \mathbf x\in \mathbf \delta (x_0).
\end{equation}
The coefficients $\hat {\beta}(\mathbf x_0)$ are obtained from the weighted least squares problem
\begin{equation}
\hat{\beta}(\mathbf x_0) =\min_{\beta} \sum_{i=1}^{n}K_H(\mathbf X_i-\mathbf x_0)\Bigl\{Y_i-\sum_{\|\mathbf j\| \le p} \beta_{j}(\mathbf x_0)(\mathbf X_i-\mathbf x_0)^{j}\Bigr\}^2,
\end{equation}
where $K_H$ is a kernel function and $H$ is the bandwidth matrix. The estimators are  
\begin{equation}
 \hat m^{(i)}(\mathbf x_0)=i!\hat\beta_{i} \quad(\|i\|\le p).
\end{equation}
We choose local polynomial regression for three reasons.
\begin{enumerate}
\item \textbf{Local approximation.}
      The assumption
      $m(\mathbf x)\approx \sum_{\|\mathbf j\| \le p} \beta_{j}(\mathbf x_0)(\mathbf x-\mathbf x_0)^{j}$ is only local and therefore matches the local cylindrification approximation.
\item \textbf{Derivative estimation.}
      Once the local polynomial is fitted, each coefficient corresponds to a partial derivative. Hence we can estimate $\partial_s r$ and $\partial_{t}r$ directly.
\item \textbf{Good asymptotic behavior.}
With suitable bandwidth and sufficiently large sample size, local polynomial regression gives consistent estimators for the regression function and for its derivatives \cite{fan2018local}. Moreover, the suitably scaled estimation errors are asymptotically normal \cite{masry1996multivariate}. This is convenient for later statistical inference.
\end{enumerate}
Note that lemma\ref{lemma3} provides the theoretical justification for using only first-order local polynomial regression: the piecewise first-order Taylor planes already yield the required local cylindrification approximation.
\subsection{Partial derivative estimation}
Let $X(t)$ be a zero-mean and second-order integrable stochastic process, and its covariance is
\begin{equation}
    r(s,t)=E(X_sX_t).
\end{equation}
Let $\{X_k(t)\}_{k=1}^N$ be $N$ independent sample paths. For the $k$ th path define $Y_k(s,t)=X_k(s)X_k(t)$. Then we have a regression model
\begin{equation}
    Y_k(s,t)=r(s,t)+\epsilon_k(s,t),
\end{equation}
where $E(\epsilon(s,t)_k|(s,t))=0$, $Var(\epsilon(s,t)_k|(s,t))=\sigma^2<\infty$ and $E(\epsilon_k^4)<\infty$. We use the two-dimensional Epanechnikov kernel \cite{fan2018local}
\begin{equation}
    K_h=\frac{9}{16}(1-u^2)(1-v^2)\quad |u|\leq 1,|v|\leq 1,
\end{equation}
where $u=\frac{s-s_0}{h}$, $v=\frac{t-t_0}{h}$ and $h$ is the bandwidth. We assume a fixed uniform grid.
\newline
We expand $r(s,t)$ at $(s_0, t_0)$ as
\begin{equation}
    r(s,t)\approx \beta_0 +\beta_1(s-s_0)+\beta_2(t-t_0),
\end{equation}
which is enough because we only need $r_s$ and $r_t$. On this grid, let $\{(s_i,t_i)\}_{i=1}^n$ denote the design points in $T\times T$.
For each $(s_i,t_i)$ we define the empirical covariance
\begin{equation}
    \begin{split}
         \hat r(s_i,t_i)&=\frac{1}{N}\sum_{k=1}^N Y_k(s_i,t_i)\\
         &=\frac{1}{N}\sum_{k=1}^N X_k(s_i)X_k(t_i).
    \end{split}
\end{equation}
In the regression notation we write $(X_i,Y_i)=((s_i,t_i),\hat r_i)$.The local linear estimator at $(s_0,t_0)$ minimizes
\begin{equation}
    \mathcal L(\beta) = \sum_{i=1}^n K_h(s_i-s_0, t_i-t_0)\,(Y_i - \beta_0 - \beta_1(s_i-s_0) - \beta_2(t_i-t_0))^2.
\end{equation}
In matrix form
\begin{equation}
\mathcal L(\beta)=(\mathbf Y-X\beta)^{\top}W(\mathbf Y-X\beta),
\end{equation}
where $\mathbf Y$ is the response vector, $\mathbf W$ is the weight matrix and 
\begin{equation}
X=
\begin{pmatrix}
1 & s_1-s_0 & t_1-t_0\\
\vdots & \vdots & \vdots\\
1 & s_n-s_0 & t_n-t_0
\end{pmatrix}.
\end{equation}
For sample points $(s, t, Y(s,t))$ with $ s,t\in [s_0-h,s_0+h]\times[t_0-h,t_0+h]$ the least squares estimator is
\begin{equation}
\hat{\beta}=(X^{\top} W X)^{-1}X^{\top} W\mathbf Y.
\end{equation}
Thus we obtain
\begin{equation}
   \hat r=\hat\beta_0,\quad \hat r_s=\hat\beta_1,\quad \hat r_t=\hat\beta_2.
\end{equation}
Note that a uniform rectangular grid is perfectly suitable for local cylindrification.
\subsection{Asymptotic theory}
We now show that $\hat{\beta}$ is a consistent estimator.
\begin{lemma}\label{lemma4}
    $\hat{\beta}$ is a consistent estimator of $\beta$ provided that
    \begin{enumerate}
        \item $h\to 0$ and $nh^{4}\to\infty$,
        \item $r$ has continuous and bounded second order partial derivatives at $(s_0, t_0)$.
    \end{enumerate}
\end{lemma}
\begin{proof}
The idea is to use the Chebyshev inequality. We start from
\begin{equation}\label{e1}
    \hat{\beta}=(X^{\top} W X)^{-1}X^{\top} W(r+\epsilon)=(X^{\top} W X)^{-1}X^{\top} Wr+(X^{\top} W X)^{-1}X^{\top} W\epsilon,
\end{equation}
so
\begin{equation}
    E(\hat{\beta})=(X^{\top} W X)^{-1}X^{\top} Wr
\end{equation}
and
\begin{equation}
    Var(\hat{\beta})=\sigma^{2}(X^{\top}WX)^{-1}X^{\top}W^{2}X(X^{\top}WX)^{-1}.
\end{equation}
Using the behaviour of $h$ and $n$ together with the properties of the kernel, we can prove separately that
\begin{equation}
\begin{split}
        &(X^{\top} W X)^{-1}X^{\top} Wr \to \beta,\\
    &\sigma^{2}(X^{\top}WX)^{-1}X^{\top}W^{2}X(X^{\top}WX)^{-1} \to 0.
\end{split}
\end{equation}
Take a second order Taylor expansion of $r(s_i,t_i)$
\begin{equation}
    r(s_i,t_i)=\beta Z_i^{\top}+h^2R_i,
\end{equation}
where $h^2R_i$ is the second-order remainder and
\begin{equation}
    X=
    \begin{bmatrix}
        Z_1\\
        \vdots\\
        Z_n
    \end{bmatrix}.
\end{equation}
Substitute into Eq. (61)
\begin{equation}
\begin{split}
    (X^{\top} W X)^{-1}X^{\top} Wr&=(X^{\top} W X)^{-1}X^{\top} W\beta X+(X^{\top} W X)^{-1}X^{\top} Wh^2R\\
    &=\beta+(X^{\top} W X)^{-1}X^{\top} Wh^2R.
\end{split}
\end{equation}
The magnitude of $(X^{\top} W X)^{-1} X^{\top} W h^2 R$ depends on the sample size and on the bandwidth under a fixed sampling density. Denote by $m$ the number of sample points inside the support of the kernel
\begin{equation}
    m \asymp nh^2.
\end{equation}
where $n$ is the total number of design points and $h^2$ is the area of the window. If the window area is fixed, then the number of sample points depends only on $n$. Moreover
\begin{equation}
    |s-s_0|^p<h^p,\quad |t-t_0|^q<h^q,
\end{equation}
thus
\begin{equation}
    |s-s_0||t-t_0|\leq h^{p+q}.
\end{equation}
For every element of $(X^{\top} W X)_{pq}$ we have
\begin{equation}
\begin{split}
     \Bigl|\sum_{i=1}^{m} K_i (s_i - s_0)^p (t_i - t_0)^q\Bigr| &< \sum_{i=1}^{m} h^{p+q} =mh^{p+q}\asymp n h^{p+q+2}.
\end{split}
\end{equation}
where $p+q\in\{0,2\}$. We find
\begin{equation}
    (X^{\top} W X)_{pq}\asymp nh^{p+q+2}.
\end{equation}
Similarly since $|R|<C$ we have
\begin{equation}
    h^2X^{\top} WR \asymp nh^{\frac{p+q}{2}+4}.
\end{equation}
Hence
\begin{equation}
    (X^{\top} W X)^{-1}X^{\top} Wh^2R \asymp nh^{\frac{p+q}{2}+4} \frac{1}{n h^{p+q+2}}=h^{2-\frac{p+q}{2}},
\end{equation}
and
\begin{equation}
    \sigma^{2}(X^{\top}WX)^{-1}X^{\top}W^{2}X(X^{\top}WX)^{-1} \asymp\frac{1}{n h^{p+q+2}}.
\end{equation}
Here, we focus on the case p+q=2. Hence $h\to 0$ and $nh^{4}\to\infty$
\begin{eqnarray}
    Var(\hat{\beta})\to 0,\\
    E(\hat{\beta})-\beta \to 0.
\end{eqnarray}
By Chebyshev's inequality, we have
\begin{equation}
    \lim_{h\to 0,nh^{4}\to\infty}P\{|\hat{\beta}-\beta|\geq \xi \}\leq \frac{Var(\hat{\beta})}{\xi^2}=0.
\end{equation}
\end{proof}
\begin{theorem}
    Under Lemma \ref{lemma4}, the estimator $\hat{J}=\hat{r}_s+\hat{r}_t$ is a consistent estimator of $J=r_s+r_t$.
\end{theorem}
\begin{proof}
    Lemma \ref{lemma4} shows that $\hat{\beta}$ is a consistent estimator. Hence
    \begin{equation}
        (\hat r_s,\hat r_t)\xrightarrow{p}(r_s,r_t)
    \end{equation}
    and since
    \begin{equation}
    \hat J=\hat{r}_s+\hat{r}_t
    \end{equation}
    is a continuous map the continuous mapping theorem \cite{billingsley2013convergence} gives that $\hat J$ is a consistent estimator of $J$.
\end{proof}
We now consider the bias asymptotic distribution.
\begin{theorem}
  Assume Lemma~\ref{lemma4} and $nh^6\to 0$. Then
  \begin{equation}
       \sqrt{nh^4}(\hat  J- J)\xrightarrow{d} N(0, \nabla J(\beta)^{\!\top}\Sigma \nabla J(\beta))
  \end{equation}
  where $\nabla J(\beta)=(0,1,1)$ and $\Sigma$ is a symmetric positive definite matrix.
\end{theorem}
\begin{proof}
From Eq. (66), we have
\begin{equation}
\hat\beta-\beta=(X^\top W X)^{-1}X^\top W\epsilon+(X^\top W X)^{-1}X^\top W h^{2}R.
\end{equation}
Write
\begin{equation}
    S=(X^\top W X)^{-1}X^\top W\epsilon ,\quad U=(X^\top W X)^{-1}X^\top W h^{2}R.
\end{equation}
From Eq. (66) and from $nh^6\to 0$, we have
\begin{equation}
    \sqrt{nh^4} U\asymp\sqrt{nh^6} \to 0.
\end{equation}
The random term $S$ can be written as
\begin{equation}
    S = \sum_{i=1}^{n} a_{n,i}\varepsilon_i,\quad a_{n,i} = (X^{\top}WX)^{-1} X_i w_i,
\end{equation}
where $X_i$ is the $i$th row of $X$ and $w_i$ is the kernel weight at $(s_i,t_i)$. When $h\to 0$ and $nh^{4}\to\infty$ the standard arguments for two dimensional local polynomial regression imply
\begin{equation}
    \begin{split}
    &\sup_{1\le i\le n} \| \sqrt{nh^{4}}a_{n,i} \|\asymp \frac{1}{ \sqrt{nh^{4}}} \to 0,\\
    &\sum_{i=1}^{n} nh^{4}\, a_{n,i} a_{n,i}^{\top}\asymp nh^{4}\frac{1}{nh^{4}} \rightarrow\ \Sigma
    \end{split}
\end{equation}
for some symmetric positive definite matrix $\Sigma$. By a weighted central limit theorem \cite{van2000asymptotic}
\begin{equation}
    \sum_{i=1}^n\sqrt{nh^{4}}a_{n,i}=S\xrightarrow{d}N(0,\Sigma).
\end{equation}
Combine the two parts
\begin{equation}
\sqrt{nh^{4}}(\hat\beta - \beta)=\sqrt{nh^{4}}S+\sqrt{nh^{4}}U \xrightarrow{d} N(0,\Sigma).
\end{equation}
Since $J(\beta)$ is continuous and
\begin{equation}
    \nabla J(\beta)=(0,1,1)\neq \mathbf{0},
\end{equation}
the delta method \cite{van2000asymptotic} yields
\begin{equation}
    \sqrt{nh^4}(\hat  J- J)\xrightarrow{d} N(0, \nabla J(\beta)^{\!\top}\Sigma \nabla J(\beta)).
\end{equation}
\end{proof}
We have thus shown that with suitable choices of bandwidth $h$ and sample size $N$, the statistic $\hat J$ is consistent and asymptotically normal. We can now apply it to some typical stochastic dynamical models.

\section{Simulation}
\subsection{SDOF}
Consider the single degree of freedom model
\begin{equation}
    m\ddot{X}+c\dot{X}+kX=\xi(t)\quad X(0)=X_0,\dot{X}(0)=\dot{X}_0,
\end{equation}
where $m$ is the mass, $k$ is the spring constant, $c$ is the damping coefficient and $\xi(t)$ is a white noise. Known results tell us that for the covariance function $r(s, t)$, the response $X$ becomes WSS when $s$ and $t$ go to infinity while $|s-t|$ remains bounded. Hence, we can use this model to check our method.
\newline
To guarantee consistency and asymptotic normality of the estimator, the bandwidth $h$ and  $n$ must satisfy
\begin{itemize}
    \item $h\to 0$
    \item $nh^4\to\infty$
    \item $ nh^6\to0 $
\end{itemize}
Hence
\begin{equation}
    h = C n^{-a}, \quad \tfrac{1}{4} > a > \tfrac{1}{6}.
\end{equation}
The parameter values used in the simulation are listed in Table~\ref{sdof_params}.
\begin{table}[htbp]
\centering
\small
\caption{Simulation parameters for the SDOF model ($N = 10000$ realizations).}
\label{sdof_params}
\renewcommand{\arraystretch}{1.15}
\begin{tabular}{
    @{} l
    l
    c
    l
    S[table-format = 5.3]
    @{}
}
\toprule
\textbf{Category} & \textbf{Type} & \textbf{Symbol} & \textbf{Definition} & {\textbf{Value}} \\ 
\midrule
\multirow{5}{*}{System} 
  & Mass                        & $m$          & given                                              & 1 \\ 
  & Spring constant             & $k$          & $\omega_n=\sqrt{k/m}=2\ \mathrm{rad\,s^{-1}}$      & 4 \\ 
  & Damping coefficient         & $c$          & $c = 2\zeta\sqrt{mk}$,\; $\zeta=0.05$              & 0.2 \\ 
  & Damping ratio               & $\zeta$      & lightly damped                                     & 0.05 \\ 
  & Noise intensity             & $D$          & the intensity of $\xi(t)$                          & 1 \\
\midrule
\multirow{3}{*}{Discretizations} 
  & Time step                   & $\Delta t$   & Euler–Maruyama step                                & 0.005 \\ 
  & Total duration              & $T$          & $\ge 10\tau_d,\;\tau_d = 1/(\zeta\omega_n)$        & 200  \\ 
  & Sample points               & $n$        & $T/\Delta t + 1$                                   & 40001 \\ 
\midrule
\multirow{4}{*}{LPR bandwidth} 
  & Realizations                & $N$          & given                                              & 10000 \\ 
  & Continuous bandwidth        & $h$          & $h = C n^{-1/5},\;C = 1$                           & 0.12 \\ 
  & Window size                 & $2L\!+\!1$   & $L = \lceil h/\Delta t \rceil$                     & 50 \\ 
\bottomrule
\end{tabular}
\end{table}
\begin{figure}[H]
    \centering
    \includegraphics[width=0.8\linewidth]{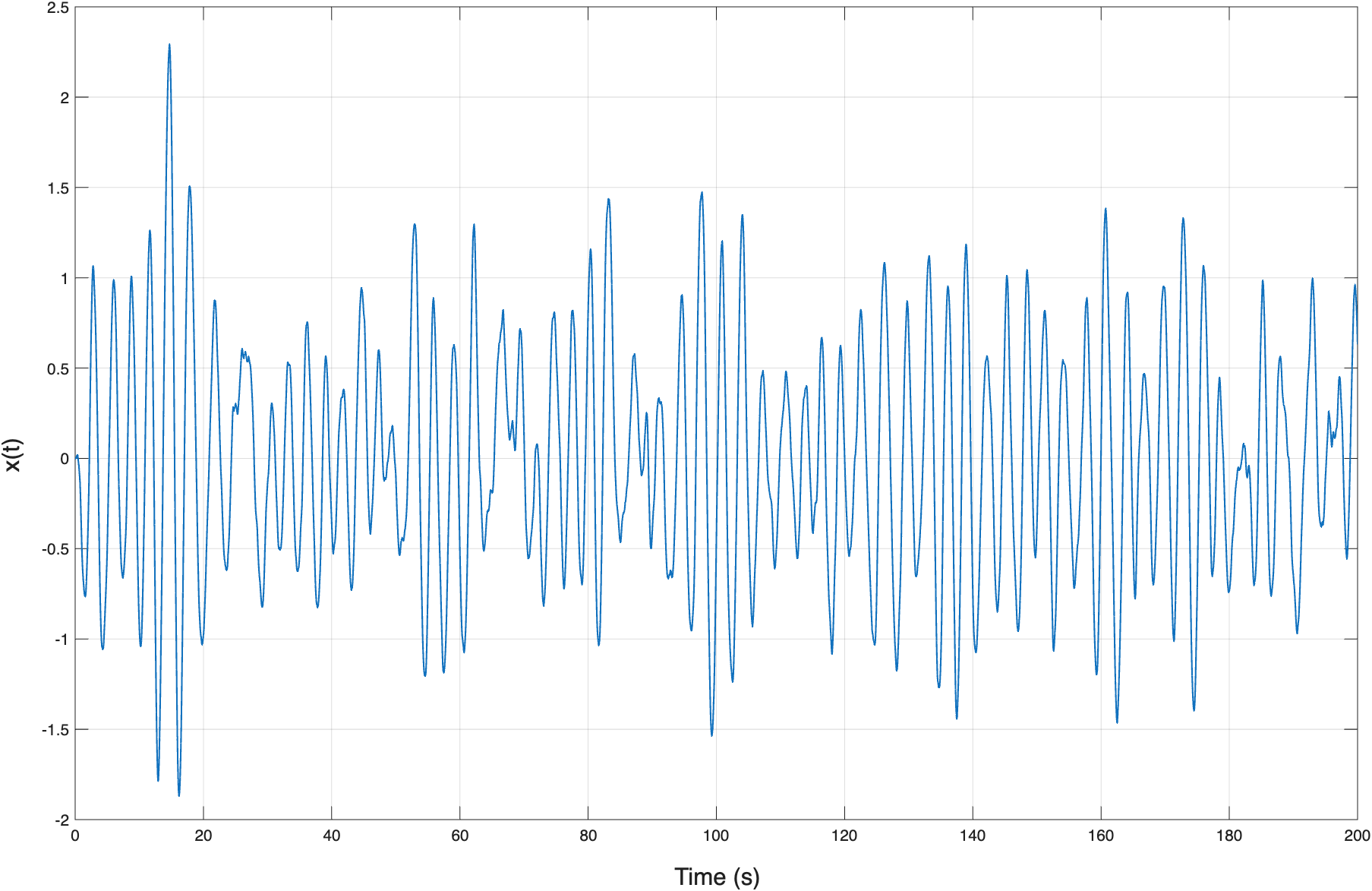}
    \caption{A sample path of SDOF Response.}
    \label{Figure1}
\end{figure}
Figure~\ref{Figure1} shows one sample path of $X$ over the time interval from zero to two hundred seconds. We use the statistic introduced in Section 3. We partition the time interval $[0, T]$ into subintervals $\{\tau_i\}_{i=1}^M $ such that $\bigcup_{i=1}^M \tau_i = [0, T]$. On each $\tau_i$, we consider the restriction of the process, $X_{\tau_i}$. Corresponding to the local cylindrical approximation, we take the patch $\Omega_i = \tau_i \times \tau_i$. In practice, it suffices to choose evaluation points $(s_i, t_i) \in \Omega_i$; for convenience, we select the diagonal representatives $s_i = t_i$ to compute the ruling direction of $\Omega_i$.For each $\tau_i$ we evaluate the statistic at the diagonal point $(s_0,t_0)=(t_i,t_i)$ and write
 \begin{equation}
      \hat J(\tau_i) = \hat r_s(t_i,t_i)+\hat r_t(t_i,t_i). 
 \end{equation}
Although we report the statistic at the diagonal point $(t_i,t_i)$ for convenience, it is estimated via a local polynomial fit on the covariance surface using nearby values in $\Omega_i$, rather than by pointwise differentiation of the diagonal. Of course, we could introduce small offsets, but under our theoretical framework, this is unnecessary and would instead sacrifice computational and inferential convenience.
\begin{figure}[H]
    \centering
    \includegraphics[width=0.8\linewidth]{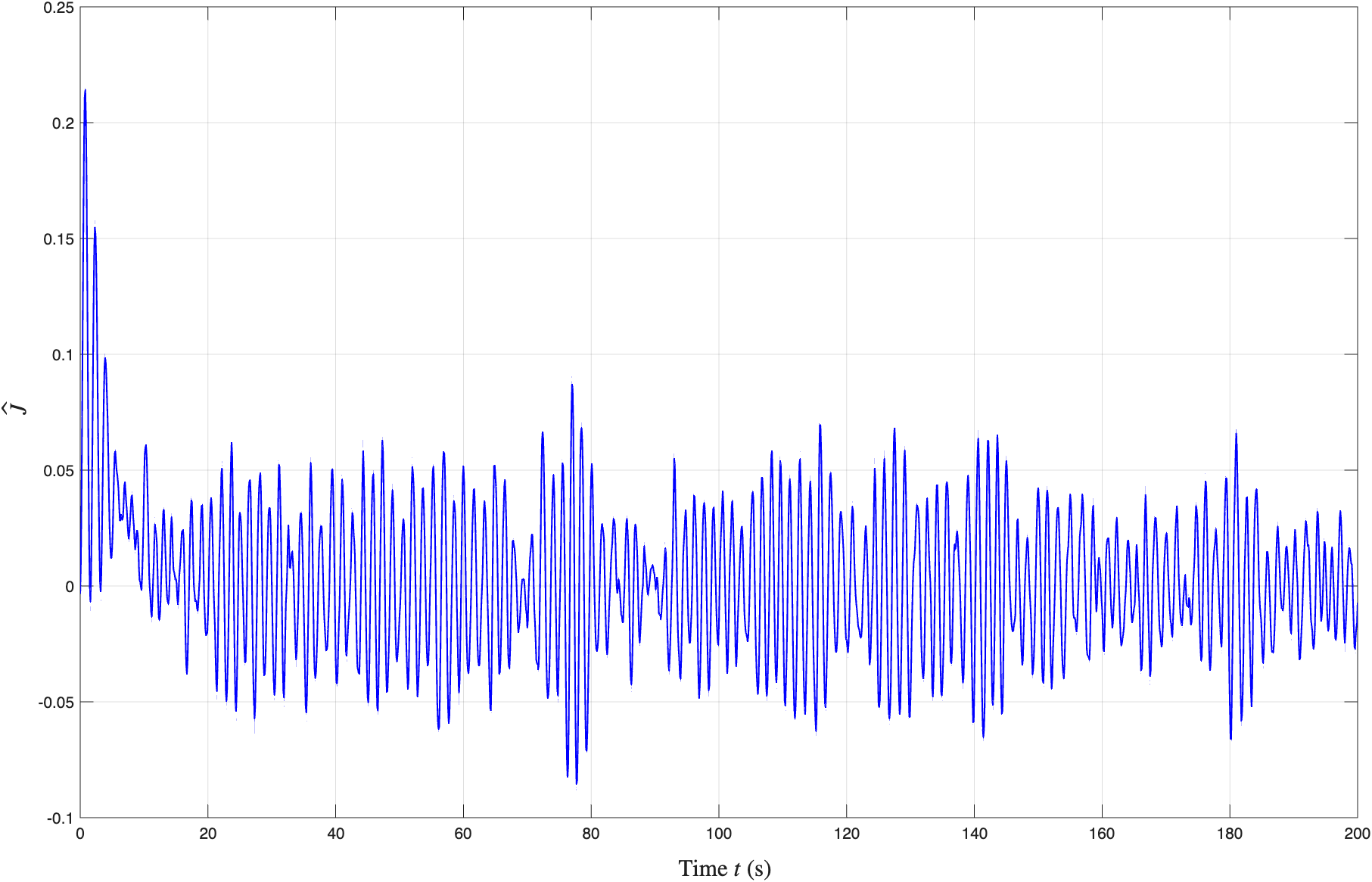}
    \caption{Time series of the statistic $\hat J$.}
    \label{Figure2}
\end{figure}
As we can see from Figure~\ref{Figure2}, the statistic $J$ starts near $0.2$. Then it decreases and is close to zero at about fifteen seconds. After that, it stays around zero. In terms of WSS, this means that the process is non-stationary at the beginning, then goes through a transition stage, and finally becomes stationary. Part of the fluctuation in Figure~\ref{Figure2} is due to finite sample size. As shown in Figure~\ref{Figure_compsize}, the variance of $\hat J$ becomes smaller when the sample size increases. However, the time points where WSS holds or fails are stable across the sample sizes considered.
\begin{figure}[H]
    \centering
    \includegraphics[width=0.8\linewidth]{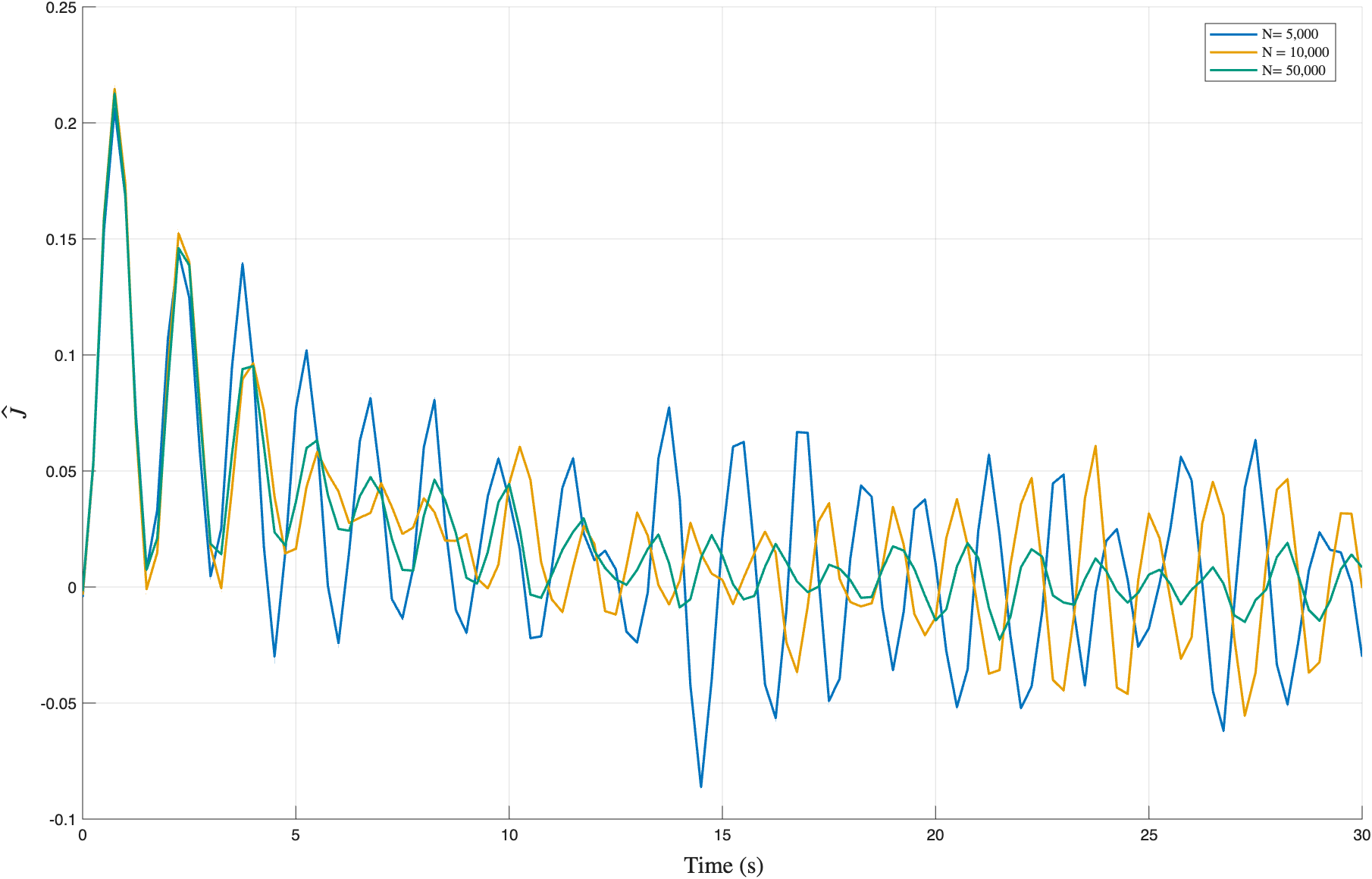}
    \caption{Time series of $\hat J$ over $[0,30]$ seconds for different sample sizes.}
    \label{Figure_compsize}
\end{figure}
Using the asymptotic normality obtained in Section 3, we set up a hypothesis test. Since
\begin{equation}
     \sqrt{nh^4}(\hat  J- J)\xrightarrow{d} N(0, \nabla J(\beta)^{\!\top}\Sigma \nabla J(\beta)),
\end{equation}
Testing is reduced to checking whether the mean of $\hat J$ is zero. When $E(\hat J) = 0$, we have $J = 0$ and hence the process is WSS. We write
\begin{equation}
    H_0:\ E(\hat J)=0 \quad \text{vs}\quad H_1:\ E(\hat J)\neq 0.
\end{equation}
For a fixed point $(s_0, t_0)$, we split the $N$ sample paths into $G$ groups and estimate $\hat{J}$ for every group with $N/G$ samples. We thus get $G$ values of $\hat{J}$ for this point, denoted by $\hat{J}_i,i=1,2\ldots G$. The estimators of the mean and variance are
\begin{equation}
    \overline{J}=\frac{1}{G}\sum_{i=1}^G\hat{J}_i ,\quad S=\frac{1}{G-1}\sum_{i=1}^G(\hat{J}_i-\overline{J})^2.
\end{equation}
The test statistic is
\begin{equation}
    T=\frac{ \overline{J}}{ S/\sqrt{G}}.
\end{equation}
Under $H_0$ we have $T\sim t_{G-1}$. Therefore, for a chosen confidence level $\alpha$, we can decide whether WSS holds at each time point. The final result is shown in Figure~\ref{Figure4}.
\begin{figure}[H]
    \centering
    \includegraphics[width=0.8\linewidth]{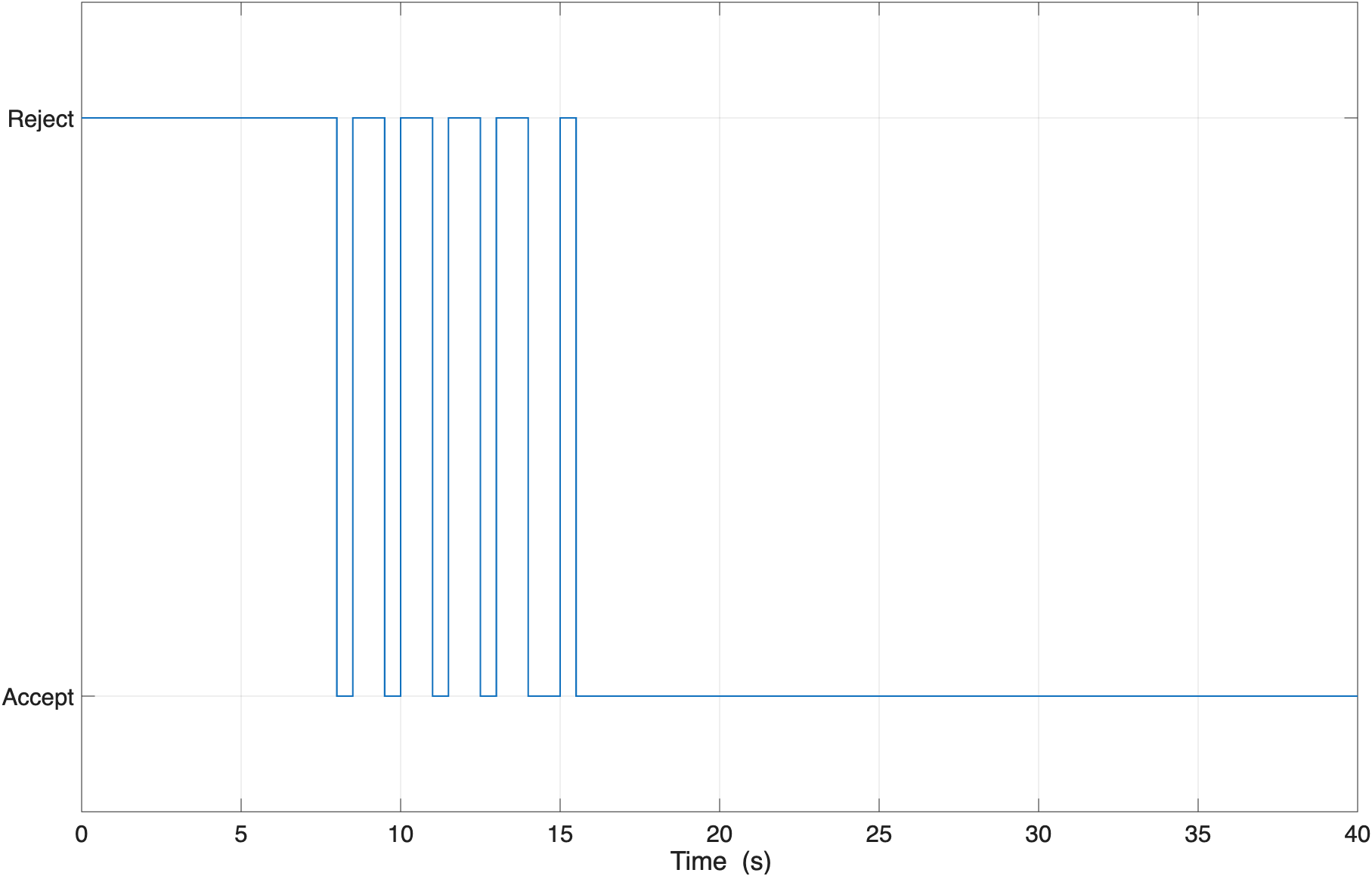}
    \caption{Two sided test with confidence level $0.95$ for the mean of $\hat{J}$.}
    \label{Figure4}
\end{figure}
From Figure~\ref{Figure4} we see that the process is clearly non-stationary from zero to about eight seconds. Between about eight seconds and about fifteen seconds, the stationarity result fluctuates. After about sixteen seconds, the process gradually becomes WSS. This is consistent with the known results from stochastic dynamics \cite{sun2006stochastic}.
\newline
If we apply Eq. (8.14) in \cite{sun2006stochastic}
\begin{equation}
\begin{split}
    \kappa_{xx}(t_1, t_2) &=
\frac{\pi S_0}{2m^2 \zeta \omega_n^3}
\Bigg\{
e^{-\zeta\omega_n|t_2-t_1|}
\left[
\cos\omega_d(t_1-t_2)
+ \frac{\zeta\omega_n}{\omega_d}\sin\omega_d|t_2-t_1|
\right] \\
&-e^{-\zeta\omega_n(t_1+t_2)}
\left[
\frac{\omega_n^2}{\omega_d^2}\cos\omega_d(t_2-t_1)
+ \frac{\zeta\omega_n}{\omega_d}\sin\omega_d(t_1+t_2)
+ \frac{\zeta^2 \omega_n^2}{\omega_d^2}\cos\omega_d(t_1+t_2)
\right]
\Bigg\},
\end{split}
\end{equation}
where $\omega_d = \omega_n \sqrt{1-\zeta^2}$, we see that the term
\begin{equation}
e^{-\zeta\omega_n(t_1+t_2)}
\left[
\frac{\omega_n^2}{\omega_d^2}\cos\omega_d(t_2-t_1)
+ \frac{\zeta\omega_n}{\omega_d}\sin\omega_d(t_1+t_2)
+ \frac{\zeta^2 \omega_n^2}{\omega_d^2}\cos\omega_d(t_1+t_2)
\right]
\end{equation}
is the part that prevents WSS. As time increases, this term decays due to the factor $e^{-\zeta \omega_n (t_1 + t_2)}$ and the response tends to be WSS. Therefore, examining the change of WSS is essentially the same as examining the decay of $e^{-\zeta \omega_n (t_1 + t_2)}$. When
\begin{equation}
    e^{-\zeta \omega_n (t_1 + t_2)}\leq \epsilon,
\end{equation}
we have
\begin{equation}
t_1+t_2 \ge \frac{1}{2\zeta\omega_n}\ln\frac{1}{\epsilon}.
\end{equation}
If we take $\epsilon=0.05$ and substitute the parameters, we obtain
\begin{equation}
    t_1+t_2\geq 5\ln 20\simeq 14.98.
\end{equation}
This is consistent with the result of our test.

\subsection{Duffing Oscillator}
Consider the Duffing stochastic oscillator \cite{belousov2019volterra}
\begin{equation}
    \ddot x + a\dot x + bx + cx^3 = \sigma\,\xi(t),
\end{equation}
where $\xi(t)$ is Gaussian white noise and $\sigma$ is the noise intensity. To test the method, we select four sets of parameters shown in Table~\ref{tab:duffing_stationarity_simplified}.
\begin{table}[htbp]
\centering
\small
\caption{Duffing system parameter sets and stationarity classification.}
\label{tab:duffing_stationarity_simplified}
\renewcommand{\arraystretch}{1.2}
\begin{tabular}{@{} lccccc@{}}
\toprule
\textbf{Case} & $a$ (damping) & $b$ (linear) & $c$ (cubic) & $\sigma$ (noise intensity) & \textbf{Stationary} \\
\midrule
Case 1 Stationary fast        & 0.50 & 1.0  & 1.0  & 0.20 & Yes \\
Case 2 non-stationary without damping & 0.00 & 1.0  & 1.0  & 0.20 & No  \\
Case 3 Stationary double well & 0.50 & $-1.0$ & 1.0  & 0.20 & Yes \\
Case 4 Stationary slow        & 0.05 & 1.0  & 1.0  & 0.20 & Yes \\
\bottomrule
\end{tabular}
\end{table}
Classical results \cite{lin2004probabilistic} give approximate formulas for the covariance of the response, but here our main interest is how stationarity evolves in time.
\begin{figure}[H]
    \centering
    \includegraphics[width=1\linewidth]{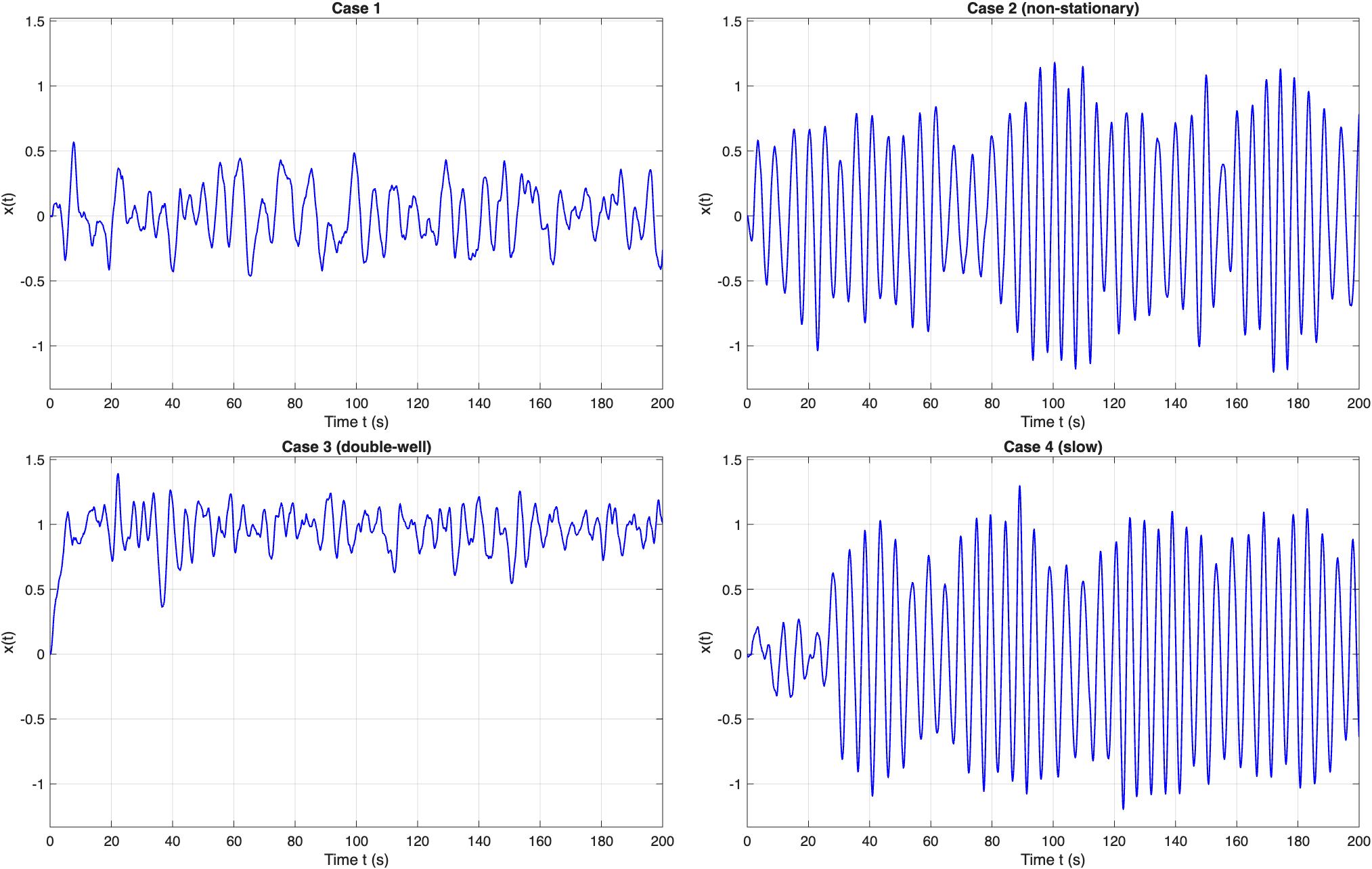}
    \caption{Time series of one trajectory of the Duffing oscillator.}
    \label{duffing_ts}
\end{figure}
As shown in Figure~\ref{duffing_ts}, changing the parameters produces very different responses. We now apply our stationarity test to these four responses.
\begin{figure}[H]
    \centering
    \includegraphics[width=1\linewidth]{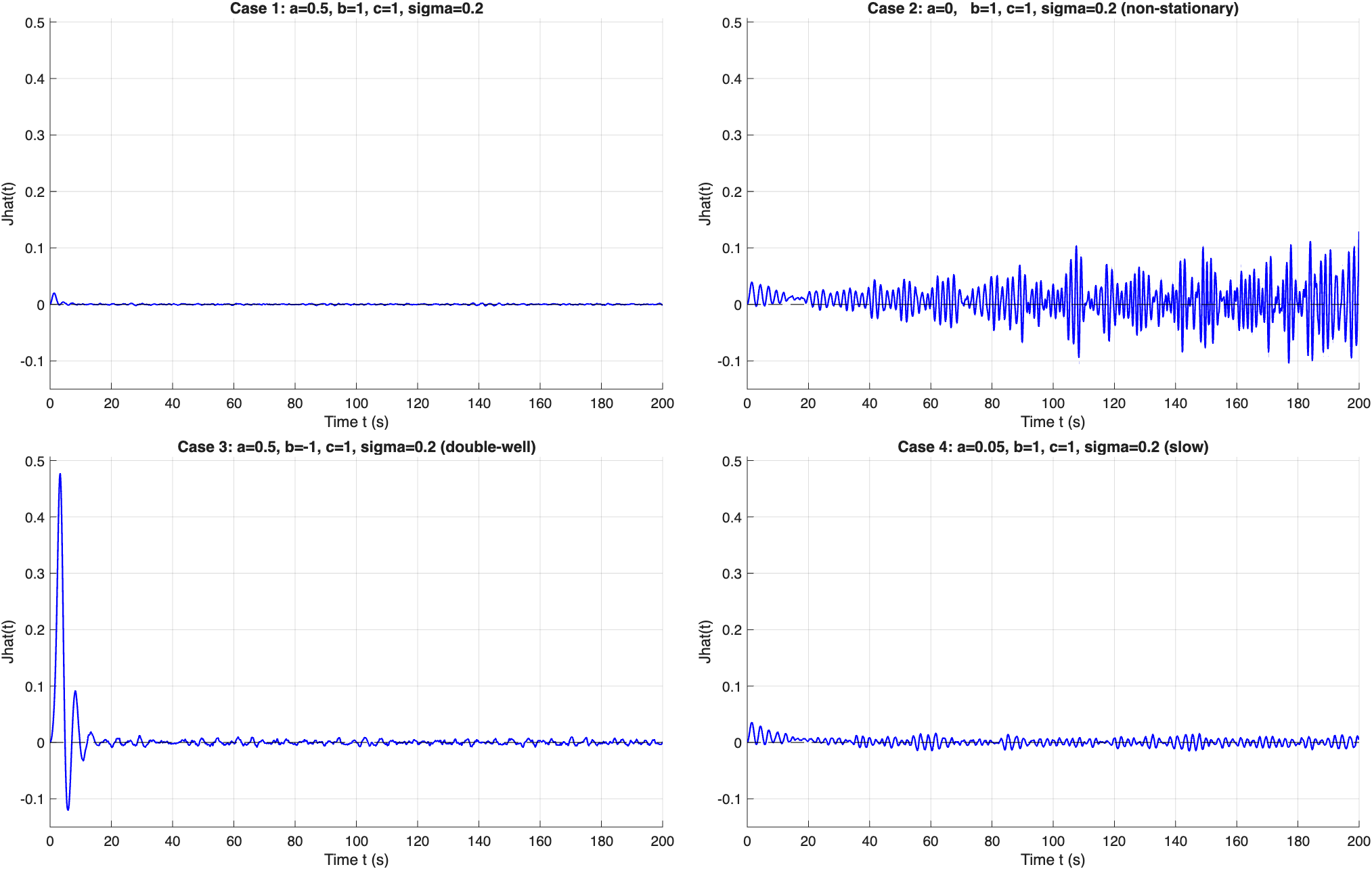}
    \caption{Test results from zero to two hundred seconds.}
    \label{duffing_J}
\end{figure}
Figure~\ref{duffing_J} shows that Case 1 is non-stationary only at the beginning and then becomes stationary very soon. Case 2 does not show stationarity. Case 3 becomes stationary but only after a period with large fluctuations. Case 4 behaves similarly to Case 1 but converges to stationarity more slowly. This can be seen more clearly from Figure~\ref{duffing_J_20}.
\begin{figure}[H]
    \centering
    \includegraphics[width=1\linewidth]{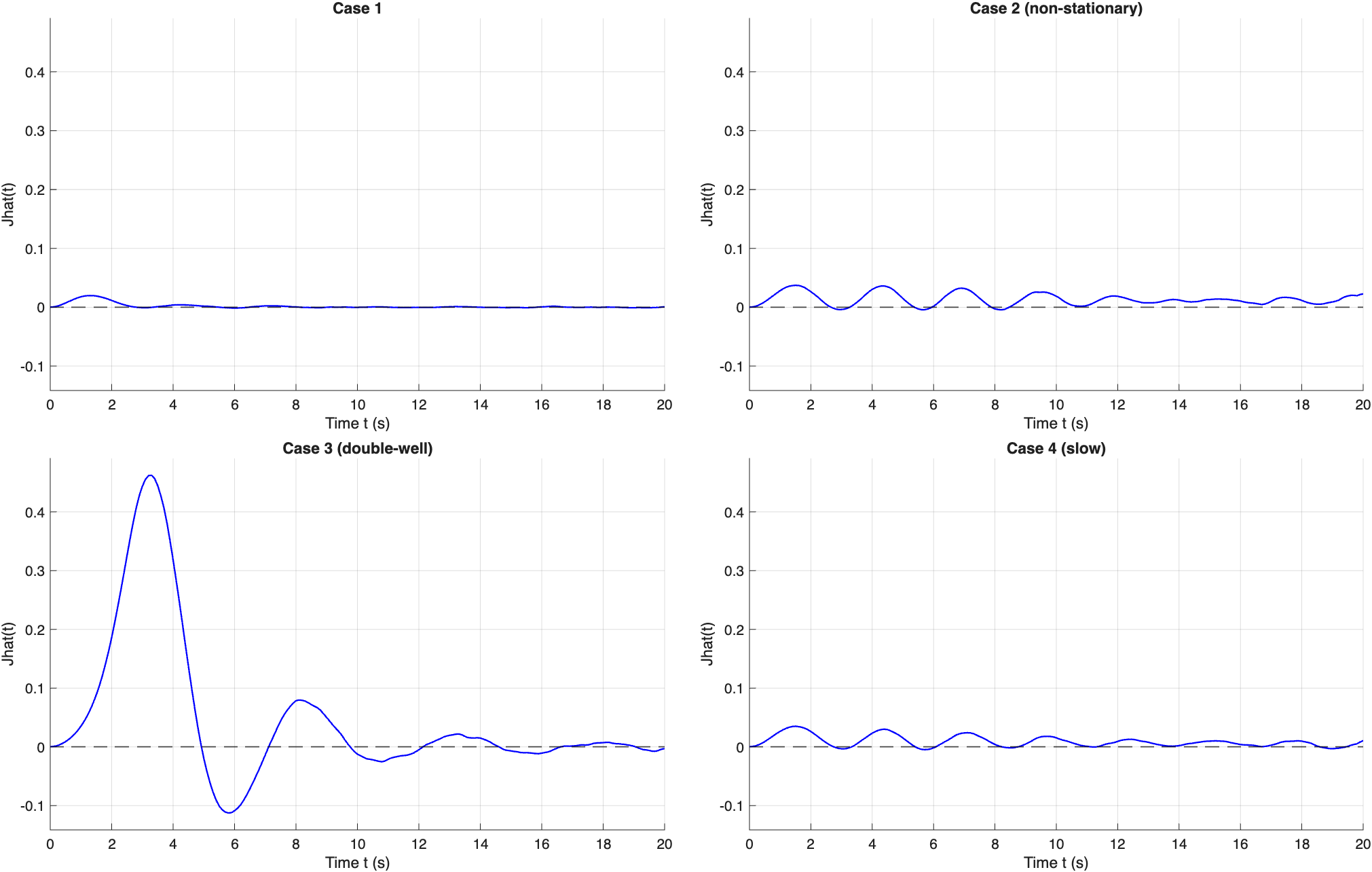}
    \caption{Test results from zero to twenty seconds.}
    \label{duffing_J_20}
\end{figure}
Although the test does not directly show the time evolution of the response, it can be combined with Fokker–Planck theory to describe how the probability density changes in time. When the probability density reaches a steady state, the response becomes strictly stationary and therefore also satisfies WSS.
\newline
For the evolution of the probability density of the Duffing oscillator under Gaussian white noise, there is a rich literature. When $a = 0$, there is no stationary probability density \cite{belousov2019volterra}. This corresponds to Case 2, and the test also reports non-stationarity. Cases 1, 3, and 4 have stationary probability densities, so after a period of time, the responses become stationary. This is in agreement with the result of our test.
\newline
We see that Case 3 has two potential wells, while Case 1 and Case 4 each have only one.
In the single well case, damping removes energy quickly, and the probability density soon collapses to one sharp peak. The response statistics then settle in a short time, and the statistic $\hat{J}(t)$ departs from zero only briefly. The difference between Case 1 and Case 4 comes from the value of a, which changes the rate of convergence \cite{PhysRevE.87.062132}.
\newline
In the double well case (Case 3), the system first stays in one well, so the probability density at early times is unimodal and asymmetric. In the long run, the density should be bimodal and almost symmetric. Before frequent transitions between the two wells occur, the process has not yet reached its stationary distribution \cite{belousov2019volterra}. This leads to a longer transient period in the statistic $\hat{J}(t).$
\subsection{Comparison}
We compare the method proposed in \cite{dette2011measure} with the method proposed in this paper. As shown in Figure~\ref{figdpv}, both methods can detect the evolution of WSS. The difference is that our method does not rely on the assumption of a locally stationary model.
\begin{figure}[H]
    \centering
    \includegraphics[width=1\linewidth]{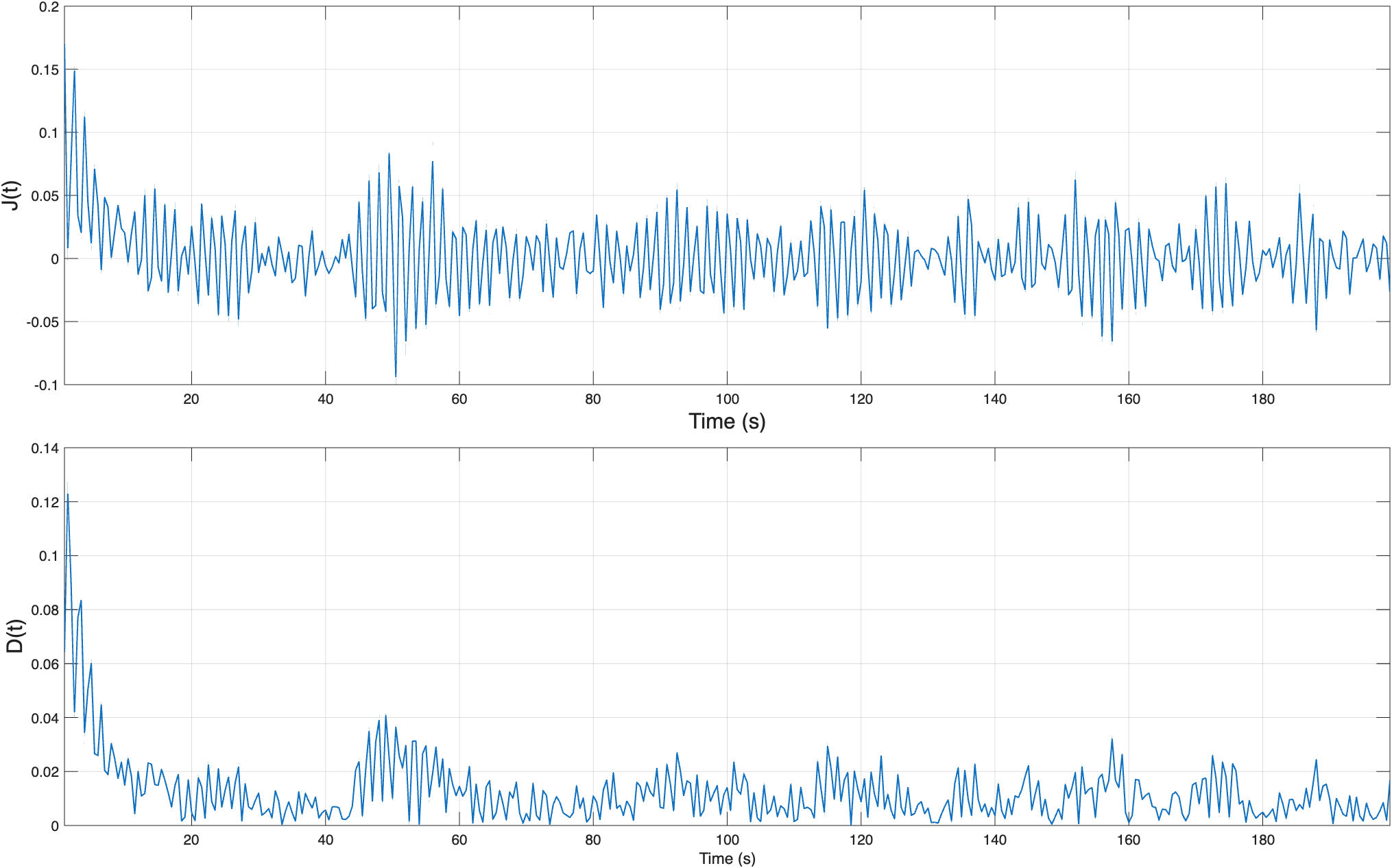}
    \caption{Comparison of  $J(t)=r_s+r_t$ and $D(t)$ for the SDOF oscillator}
    \label{figdpv}
\end{figure}
However, for a Wiener process, the DPV measure is less informative in our setting. As shown in Figure~\ref{figWiener}, J(t) has the theoretical value 1, and its estimate stays close to 1, clearly indicating a persistent departure from WSS, whereas the DPV measure does not provide a comparably clear diagnostic.
\begin{figure}[H]
    \centering
    \includegraphics[width=1\linewidth]{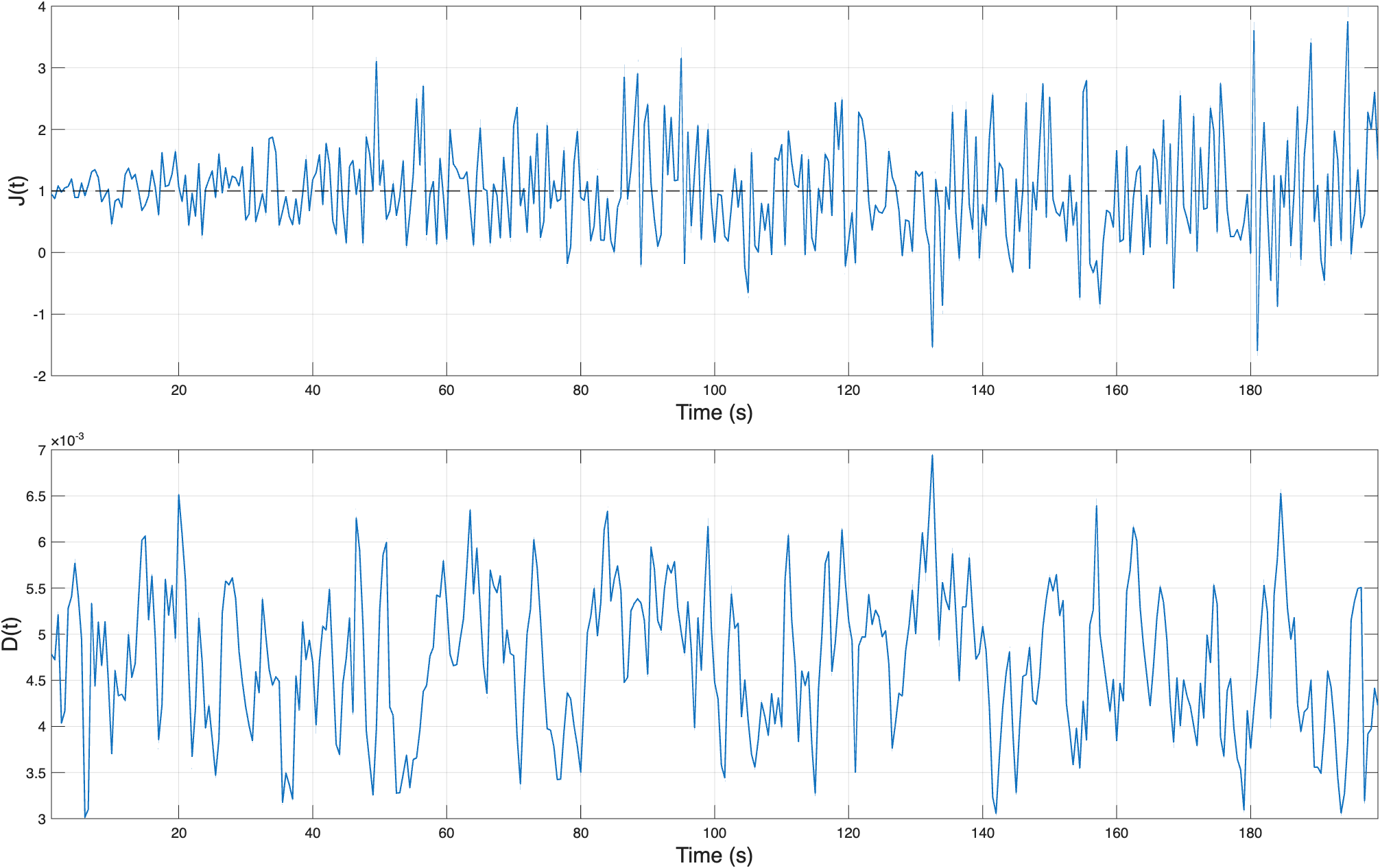}
    \caption{Comparison of  $J(t)=r_s+r_t$ and $D(t)$ for the Wiener process}
    \label{figWiener}
\end{figure}
\section{Conclusion}
In this paper, we propose a geometric method to test WSS directly. We write the classical WSS condition as a first-order condition on the covariance function. Using a local cylindrification together with local polynomial regression, we build an estimator of the derivatives and show that it is consistent and asymptotically normal. The test does not rely on any fixed parametric time series model. It works directly on the covariance surface and gives geometrically interpretable diagnostics for WSS.
\newline
The method also has several limitations. First, the current theory still needs regularity conditions on the covariance function, such as piecewise smoothness and the existence of derivatives up to some finite order. Hence, it does not yet apply to very rough covariance surfaces. Second, since local polynomial regression is performed on the two-dimensional time plane, the algorithm needs a lot of computation and requires a large number of observations. 
\newline
Future research will put more emphasis on spectral representation and data analysis within this framework. For small samples and rough covariance surfaces, it is necessary to develop more robust inference and testing methods to improve finite-sample performance while relaxing smoothness assumptions. For the representation and spectral decomposition of the covariance function, we will consider providing a spectral representation of the covariance function for non-stationary stochastic processes from an evolutionary perspective, enabling better frequency-domain analysis of the responses of stochastic dynamical systems.
\bibliographystyle{apalike} 
\bibliography{sample}

@article{priestley1965evolutionary,
  title={Evolutionary spectra and non-stationary processes},
  author={Priestley, Maurice B},
  journal={Journal of the Royal Statistical Society: Series B (Methodological)},
  volume={27},
  number={2},
  pages={204--229},
  year={1965},
  publisher={Wiley Online Library}
}

@book{fan2018local,
  title={Local polynomial modelling and its applications: monographs on statistics and applied probability 66},
  author={Fan, Jianqing},
  year={2018},
  publisher={Routledge}
}

@article{de2013derivative,
  title={Derivative estimation with local polynomial fitting},
  author={De Brabanter, Kris and De Brabanter, Joseph and Gijbels, Irene and De Moor, Bart},
  journal={Journal of Machine Learning Research},
  volume={14},
  number={1},
  pages={281--301},
  year={2013},
  publisher={MIT Press}
}

@article{masry1996multivariate,
  title={Multivariate local polynomial regression for time series: uniform strong consistency and rates},
  author={Masry, Elias},
  journal={Journal of Time Series Analysis},
  volume={17},
  number={6},
  pages={571--599},
  year={1996},
  publisher={Wiley Online Library}
}

@book{do2016differential,
  title={Differential geometry of curves and surfaces: revised and updated second edition},
  author={Do Carmo, Manfredo P},
  year={2016},
  publisher={Courier Dover Publications}
}

@book{billingsley2013convergence,
  title={Convergence of probability measures},
  author={Billingsley, Patrick},
  year={2013},
  publisher={John Wiley \& Sons}
}

@book{sun2006stochastic,
  title={Stochastic dynamics and control},
  author={Sun, Jian-Qiao},
  volume={4},
  year={2006},
  publisher={Elsevier}
}

@article{page1954continuous,
  title={Continuous inspection schemes},
  author={Page, Ewan S},
  journal={Biometrika},
  volume={41},
  number={1/2},
  pages={100--115},
  year={1954},
  publisher={JSTOR}
}

@article{dahlhaus1997fitting,
  title={Fitting time series models to nonstationary processes},
  author={Dahlhaus, Rainer},
  journal={The annals of Statistics},
  volume={25},
  number={1},
  pages={1--37},
  year={1997},
  publisher={Institute of Mathematical Statistics}
}

@article{dette2011measure,
  title={A measure of stationarity in locally stationary processes with applications to testing},
  author={Dette, Holger and Preu{\ss}, Philip and Vetter, Mathias},
  journal={Journal of the American Statistical Association},
  volume={106},
  number={495},
  pages={1113--1124},
  year={2011},
  publisher={Taylor \& Francis}
}

@article{bebendorf2003note,
  title={A note on the Poincar{\'e} inequality for convex domains},
  author={Bebendorf, Mario},
  journal={Zeitschrift f{\"u}r Analysis und ihre Anwendungen},
  volume={22},
  number={4},
  pages={751--756},
  year={2003}
}

@article{hinkley1971inference,
  title={Inference about the change-point from cumulative sum tests},
  author={Hinkley, David V},
  journal={Biometrika},
  volume={58},
  number={3},
  pages={509--523},
  year={1971},
  publisher={Oxford University Press}
}

@article{bai1998estimating,
  title={Estimating and testing linear models with multiple structural changes},
  author={Bai, Jushan and Perron, Pierre},
  journal={Econometrica},
  pages={47--78},
  year={1998},
  publisher={JSTOR}
}

@article{belousov2019volterra,
  title={Volterra-series approach to stochastic nonlinear dynamics: The Duffing oscillator driven by white noise},
  author={Belousov, Roman and Berger, Florian and Hudspeth, AJ},
  journal={Physical Review E},
  volume={99},
  number={4},
  pages={042204},
  year={2019},
  publisher={APS}
}

@article{PhysRevE.87.062132,
  title = {Stationary energy probability density of oscillators driven by a random external force},
  author = {M\'endez, Vicen\ifmmode \mbox{\c{c}}\else \c{c}\fi{} and Campos, Daniel and Horsthemke, Werner},
  journal = {Phys. Rev. E},
  volume = {87},
  issue = {6},
  pages = {062132},
  numpages = {10},
  year = {2013},
  month = {Jun},
  publisher = {American Physical Society},
  doi = {10.1103/PhysRevE.87.062132},
  url = {https://link.aps.org/doi/10.1103/PhysRevE.87.062132}
}

@book{lin2004probabilistic,
  title = {Probabilistic Structural Dynamics: {{Advanced}} Theory and Applications},
  author = {Lin, Y.K. and Cai, G.Q.},
  year = 2004,
  series = {Engineering Reference Series},
  publisher = {McGraw-Hill},
  isbn = {978-0-07-143800-1},
  lccn = {2004049068}
}

@book{van2000asymptotic,
  title={Asymptotic statistics},
  author={Van der Vaart, Aad W},
  volume={3},
  year={2000},
  publisher={Cambridge University Press}
}

@article{paparoditis2009testing,
     author = {Paparoditis, Efstathios},
     title = {Testing temporal constancy of the spectral structure of a time series},
     journal = {Bernoulli},
     volume = {15},
     number = {1},
     year = {2009},
     pages = { 1190-1221},
     language = {en},
     url = {http://dml.mathdoc.fr/item/1262962232}
}

@article{nason2013test,
  title={A test for second-order stationarity and approximate confidence intervals for localized autocovariances for locally stationary time series},
  author={Nason, Guy},
  journal={Journal of the Royal Statistical Society Series B: Statistical Methodology},
  volume={75},
  number={5},
  pages={879--904},
  year={2013},
  publisher={Oxford University Press}
}

@article{aue2020testing,
  title={Testing for stationarity of functional time series in the frequency domain},
  author={Aue, Alexander and Van Delft, Anne},
  journal={The Annals of Statistics},
  volume={48},
  number={5},
  pages={2505--2547},
  year={2020},
  publisher={JSTOR}
}

@article{van2021nonparametric,
  title={A nonparametric test for stationarity in functional time series},
  author={van Delft, Anne and Characiejus, Vaidotas and Dette, Holger},
  journal={Statistica Sinica},
  volume={31},
  number={3},
  pages={1375--1395},
  year={2021},
  publisher={JSTOR}
}

@article{yao1987approximating,
  title={Approximating the distribution of the maximum likelihood estimate of the change-point in a sequence of independent random variables},
  author={Yao, Yi-Ching},
  journal={The Annals of Statistics},
  pages={1321--1328},
  year={1987},
  publisher={JSTOR}
}

@article{andrews1993tests,
  title={Tests for parameter instability and structural change with unknown change point},
  author={Andrews, Donald WK},
  journal={Econometrica: Journal of the Econometric Society},
  pages={821--856},
  year={1993},
  publisher={JSTOR}
}

@article{einmahl2003empirical,
  title={Empirical likelihood based hypothesis testing},
  author={Einmahl, John HJ and McKeague, Ian W},
  journal={Bernoulli},
  volume={9},
  number={2},
  pages={267--290},
  year={2003},
  publisher={Bernoulli Society for Mathematical Statistics and Probability}
}

@article{casini2024change,
  title={Change-point analysis of time series with evolutionary spectra},
  author={Casini, Alessandro and Perron, Pierre},
  journal={Journal of Econometrics},
  volume={242},
  number={2},
  pages={105811},
  year={2024},
  publisher={Elsevier}
}

@article{last2008detecting,
  title={Detecting abrupt changes in a piecewise locally stationary time series},
  author={Last, Michael and Shumway, Robert},
  journal={Journal of multivariate analysis},
  volume={99},
  number={2},
  pages={191--214},
  year={2008},
  publisher={Elsevier}
}

@article{forgoston2018primer,
  title={A primer on noise-induced transitions in applied dynamical systems},
  author={Forgoston, Eric and Moore, Richard O},
  journal={SIAM Review},
  volume={60},
  number={4},
  pages={969--1009},
  year={2018},
  publisher={SIAM}
}

@book{brockwell2009time,
  title={Time series: theory and methods},
  author={Brockwell, Peter J and Davis, Richard A},
  year={2009},
  publisher={Springer science \& business media}
}

@book{liang2015random,
  title={Random vibration: mechanical, structural, and earthquake engineering applications},
  author={Liang, Zach and Lee, George C},
  year={2015},
  publisher={CRC Press}
}

@book{leoni2017first,
  title={A first course in Sobolev spaces},
  author={Leoni, Giovanni},
  year={2017},
  publisher={American Mathematical Soc.}
}

\end{document}